\documentclass[12pt,english]{article}
\usepackage{amsmath}
\usepackage{amsfonts}
\usepackage{amssymb}
\usepackage{amsthm}

\usepackage[left=25mm,top=26mm,right=26mm,bottom=30mm]{geometry}

\newtheorem{theorem}{Theorem}[section]

\newtheorem{lemma}{Lemma}[section]

\def\tire{\thinspace--\thinspace}
\let\epsilon=\varepsilon

\title{Energy Growth of Infinite Harmonic Chain under Microscopic Random
  Influence}

\author{A.A.~Lykov \thanks{Mechanics and Mathematics Faculty, Lomonosov Moscow
State University, Leninskie Gory~1, Moscow, 119991, Russia} }

\begin{document}
\maketitle


  \begin{abstract}
    Infinite harmonic chains of point particles with finite range translation invariant interaction have considered. It is assumed that the only one particle influenced by the white noise. 
We studied microscopic and macroscopic behavior of the system's energies (potential, kinetic, total) when time goes to infinity. 
We proved that under quite general condition on interaction potential the energies grow 
linearly with time on macroscopic scale, and grow as $\ln(t)$ on microscopic scale. Moreover it is turned out that the system exhibit some equipartition properties in this non equilibrium settings.
  \end{abstract}

\section{Introduction}

We consider infinite number of point particles of unit masses on $\mathbb{R}$
with formal Hamiltonian 
\[
H(q,p)=\sum_{k\in\mathbb{Z}}\frac{p_{k}^{2}}{2}+\frac{1}{2}\sum_{k,j}a(k-j)q_{k}q_{j},\quad p_{k},q_{k}\in\mathbb{R}
\]
where $q_{k}=x_{k}-k\delta$ denotes displacement of the particle
with index $k$ from the point $k\delta$ for some $\delta>0$, $p_{k}$
is the impulse of the particle $k$. We will assume that function
$a(n)$ satisfies the following three natural conditions: 
\begin{enumerate}
\item symmetry: $a(n)=a(-n)$, 
\item finite range interaction: $a(n)$ has finite support, i.e.\ there is
number $r\geqslant1$ such that $a(n)=0$ if $|n|>r$, 
\item for every $\lambda\in\mathbb{R}$ there is the bound 
\[
\omega^{2}(\lambda)=\sum_{n\in\mathbb{Z}}a(n)e^{in\lambda}\geqslant0.
\]
Next we will see that this condition guarantees that the energy $H(q,p)$
is non-negative for all $p,q\in l_{2}(\mathbb{Z})$. 
\end{enumerate}
Harmonic chain with nearest neighborhood interaction by definition
has the following Hamiltonian: 
\[
H_{C}(q,p)=\sum_{k\in\mathbb{Z}}\frac{p_{k}^{2}}{2}+\frac{\omega_{1}^{2}}{2}\sum_{k}(q_{k+1}-q_{k})^{2}+\frac{\omega_{0}^{2}}{2}\sum_{k}q_{k}^{2}
\]
for some non-negative constants $\omega_{0},\omega_{1}\geqslant0$.
It is easy to see that 
\[
\omega^{2}(\lambda)=\omega_{0}^{2}+2\omega_{1}^{2}(1-\cos\lambda)
\]
and the conditions 1--3 obviously hold.

We will always assume that conditions 1--3 are fulfilled.

Suppose that the particle with index $0$ is influenced by the white
noise. Then the equations of motion are: 
\begin{equation}
\ddot{q}_{k}=-\sum_{j}a(k-j)q_{j}+\sigma\delta_{k,0}\dot{w}_{t},\quad  k\in\mathbb{Z},\label{mainEquation}
\end{equation}
where 
\[
\delta_{k,0}=\begin{cases}
1, & k=0,\\
0, & k\ne0,
\end{cases}
\]
$\sigma>0$ and $w_{t}$ is a standard Brownian motion.

In this article we study the solution $q(t),p(t)$ of latter equations
with initial data in $l_{2}(\mathbb{Z})$. Mainly we are interested
in the energies behavior as $t\rightarrow\infty$ at the microscopic
(local) and macroscopic (global) scales. By energies we mean three
quantities: kinetic, potential and full energy.

The first and more simple objects to study are global energies. By
global we mean the energy of the whole system. We will prove that
the expectation values of the energies grow linearly with time $t$
up to terms $\bar{\bar{o}}(t)$. Moreover, the mean values of the
kinetic and potential energies grow with the same speed $\sigma^{2}/4$.
It shows that our system, in some sense, is transient and goes to
infinity. Despite this fact, we see the equipartition property at
this pure non-equilibrium case. Recall that accordingly to the equipartition
theorem (\hspace{-0.1pt}\cite{Huang}, p.\thinspace 136--138) for systems with quadratic Hamiltonian
at the equilibrium, the expected values of the kinetic and potential
energies coincide.

Next, we investigate the local energies. It turns out that on microscopic
level the equipartition property holds too. We prove that under some
additional assumptions the $n$-th particle mean kinetic and potential
energies asymptotically equal $d_{n}\ln t$ as $t\rightarrow\infty$
for some constants $d_{n}\geqslant d>0$ bounded from below by some
positive constant $d$. In general $d_{n}$ depends on the index $n$,
but if $\omega(\lambda)$ has no critical points on interval $[0,\pi]$
except $0$ and $\pi$ then $d_{n}$ does not depend on $n$. Thus
in the latter case we have the equipartition property of the energy
by particles in some asymptotic sense. The related effect was mentioned
and proved in the book \cite{Kozlov} for some class of finite linear
Hamiltonian systems with random initial conditions. The order of growth
$\ln t$ is not the same as in a finite case. Indeed, in the finite
case for linear Hamiltonian system when only few degrees of freedom
are influenced by the white noise the local energies grow like $t$
(\hspace{-0.1pt}\cite{LDisser}, p.\thinspace 68--73). Therein was also proved equipartition
property of the local energies. The slower order of growth $\ln t$
in infinite case is natural since a part of the energy goes to infinity.

Infinite harmonic chain is a quite standard object in mathematical
physics. It is used for study various physical phenomenons: convergence
to equilibrium, heat transport, hydrodynamics, etc. 
Most of the works dealing with systems are not far away from an equilibrium.
In the present article we have considered, in some sense, the opposite
situation --- our system goes to infinity, i.e.\ away from equilibrium.
The effect of the total energy's growth in the case when only few
degrees of freedom are influenced by a white noise was proved in \cite{LMM}
for finite linear Hamiltonian systems. Seminal physical papers \cite{GKT,KK}
(and reference therein) are closely related to ours. In these articles
authors have considered system with a discrete Laplacian at the right
hand side of (\ref{mainEquation}). In \cite{KK} instead of white
noise there is a periodic force $\sin\omega t$ acting on the fixed
particle with index 0 and additional pinning term. The authors studied
behavior of energies and solution as time $t$ goes to infinity.


\section{Results}

Denote $p_{k}(t)=\dot{q}_{k}(t)$. We say that sequences $q_{k}(t),\ p_{k}(t),\ k\in\mathbb{Z}$
of stochastic processes solve the equation (\ref{mainEquation}) if
they satisfy the following system of stochastic differential equations
(in It\^o sense) 
\begin{align*}
dq_{k}= & \ p_{k}dt,\\
dp_{k}= & \ -\sum_{j}a(k-j)q_{j}dt+\sigma\delta_{k,0}dw_{t},\quad k\in\mathbb{Z}.
\end{align*}

Denote the phase space of our system by $L=\{\psi=(q,p):q\in l_{2}(\mathbb{Z}),\ p\in l_{2}(\mathbb{Z})\}$.
It is evident that $L$ is the Hilbert space.

\begin{lemma} \label{EUlemma} For all $\psi\in L$ there is a unique
solution $\psi(t)$ of (\ref{mainEquation}) with initial condition
$\psi$ such that $\bold{P}(\psi(t)\in L)=1$ for all $t\geqslant0$
. \end{lemma}

By unique we mean that if $\psi'(t)$ is another solution of (\ref{mainEquation})
such that $\psi'(0)=\psi$ and $\psi'(t)\in L$ almost sure for all
$t$ then $\psi(t)$ and $\psi'(t)$ are stochastically equivalent,
i.e.\ $\bold{P}(\psi(t)=\psi'(t))=1$ for all $t\geqslant0$ .

Define kinetic and potential energies respectively: 
\[
T(t)=\frac{1}{2}\sum_{k}p_{k}^{2}(t),\quad U(t)=\frac{1}{2}\sum_{k,j}a(k-j)q_{k}(t)q_{j}(t),
\]
where $\psi(t)=(q(t),p(t))^{T}$ is a solution of (\ref{mainEquation}).
Then evidently $H(t)=H(q(t),$ $p(t))=T(t)+U(t)$.

\begin{theorem}[Global energy behavior] \label{globalEnergyBehTh}
For all initial condition $\psi(0)\in L$ the following equalities
hold 
\begin{align}
\bold{E}H(t)= & \ \frac{\sigma^{2}}{2}t+H(0),\label{EHform}\\
\bold{E}T(t)= & \ \frac{\sigma^{2}}{4}t+\bar{\bar{o}}(t),\label{ETform}\\
\bold{E}U(t)= & \ \frac{\sigma^{2}}{4}t+\bar{\bar{o}}(t),\label{EUform}
\end{align}
as $t\rightarrow\infty$. Moreover the full energy has the following
representation 
\begin{equation}
H(t)=\frac{\sigma^{2}}{2}t+H(0)+\xi_{0}(t)+\xi_{1}(t),\label{energyDecomp}
\end{equation}
where a Gaussian random process $\xi_{0}(t)$ depends on initial conditions
and can be expressed as follows 
\[
\xi_{0}(t)=\int_{0}^{t}f(s)\ dw_{s},\ 
\]
and the random process $\xi_{1}(t)$ has a representation via a multiple
It\^o integral 
\[
\xi_{1}(t)=\int_{0}^{t}\left(\int_{0}^{s_{1}}h(s_{1}-s_{2})dw_{s_{2}}\right)dw_{s_{1}},\ 
\]
where 
\[
h(t)=\frac{\sigma^{2}}{2\pi}\int_{0}^{2\pi}\cos(t\omega(\lambda))d\lambda,
\]
\[
f(t)=\frac{\sigma}{2\pi}\int_{0}^{2\pi}-Q_{0}(\lambda)\omega(\lambda)\sin(t\omega(\lambda))+P_{0}(\lambda)\cos(t\omega(\lambda))d\lambda,
\]
\[
Q_{0}(\lambda)=\sum_{n}q_{n}(0)e^{in\lambda},\quad P_{0}(\lambda)=\sum_{n}p_{n}(0)e^{in\lambda}.
\]
\end{theorem}

If the initial conditions are zero, then as a consequence of presentation
(\ref{energyDecomp}), we can write the variance of the full energy:
\begin{equation}
\bold{D}H(t)=\int_{0}^{t}(t-s)h^{2}(s)ds.\label{HvarFor}
\end{equation}
Later on we will prove that the following formula holds: 
\begin{equation}
\bold{D}H(t)=\bar{\bar{o}}(t^{2}),\ \mbox{as}\ t\rightarrow\infty.\label{HvarAsymp}
\end{equation}

We see that the mean energies of the whole system grows linearly with
time. Next we will see that locally the energies grow like logarithm
of time $t$. Additionally we should note that the mean kinetic and
potential energies grow with the same rate: $\sigma^{2}/4$. It turns
out that locally this picture is the same, i.e.\ the leading asymptotic
term in the corresponding mean energies coincides. To formulate the
corresponding result we need more assumptions.

We will suppose that

\begin{itemize}
\item[{\bf A1)}] $\omega(\lambda)$ is strictly greater than zero: 
\[
\omega(\lambda)>0
\]
for all $\lambda\in\mathbb{R}$.
\item[{\bf A2)}]  each critical point of $\omega(\lambda)$ is non-degenerate, it
means that if $\omega'(\lambda)=0$ for some $\lambda$ then $\omega''(\lambda)\ne0$.
\end{itemize}

Since $a(n)$ is symmetric and has finite support, we can write 
\[
\omega^{2}(\lambda)=a(0)+2\sum_{n=1}^{r}a(n)\cos(n\lambda).
\]
Thus points $0$ and $\pi$ are critical and $\omega(\lambda)$ has
a finite number of critical points on interval $(0,\pi)$. Let us
denote their by $\lambda_{1},\ldots,\lambda_{m}$ and put by definition
$\lambda_{0}=0,\ \lambda_{m+1}=\pi$.

\begin{itemize}
\item[{\bf A3)}] The third assumption is $\omega(\lambda_{j})\ne\omega(\lambda_{i})$
for all $i\ne j$.
\end{itemize}

Let us introduce local energies 
\[
T_{n}(t)=\frac{p_{n}^{2}(t)}{2},\quad  U_{n}(t)=\frac{1}{2}\sum_{j}a(n-j)q_{n}(t)q_{j}(t),\quad  H_{n}=T_{n}+U_{n},
\]
kinetic, potential and full energy respectively.
\begin{theorem}[Local energy behavior]
\label{locEnergyTheorem} Suppose assumptions A1, A2, A3 are fulfilled. Then
for all $n\in\mathbb{Z}$ the following equalities hold: 
\begin{align}
\bold{E}T_{n}(t)= & \ d_{n}\ln(t)+O(1),\label{kineticEnergyAsym}\\
\bold{E}U_{n}(t)= & \ d_{n}\ln(t)+O(1),\label{potentialEnergyAsym}
\end{align}
where we denote 
\begin{equation}
d_{n}=\frac{\sigma^{2}}{8\pi}\sum_{k=0}^{m+1}\frac{\cos^{2}n\lambda_{k}}{|\omega''(\lambda_{k})|}\chi_{k},\quad\chi_{k}=\begin{cases}
4, & k=1,\ldots,m\\
1, & k\in\{0,m+1\}
\end{cases}.\label{dneqdef}
\end{equation}
\end{theorem}

We emphasize that $\inf_{n}d_{n}>0$. Indeed, since $\lambda_{0}=0,\ \lambda_{m+1}=\pi$,
we obtain the bound: 
\[
d_{n}\geqslant\frac{\sigma^{2}}{8\pi}\left(\frac{1}{|\omega''(0)|}+\frac{1}{|\omega''(\pi)|}\right)>0.
\]

\section{Proofs}

\subsection{Existence and uniqueness lemma \ref{EUlemma}}

First we prove the existence and uniqueness lemma \ref{EUlemma}.
Let us rewrite the system (\ref{mainEquation}) in a matrix form:
\begin{equation}
\dot{\psi}=A\psi+\sigma g_{0}\dot{w}_{t},\label{mainEqMatrixForm}
\end{equation}
where a matrix $A$ is given by the formula (in $q,p$ decomposition
of the phase space): 
\begin{equation}
A=\left(\begin{array}{cc}
0 & I\\
-V & 0
\end{array}\right),\label{defA}
\end{equation}
$I$ is the unit matrix, $V_{i,j}=a(i-j)$ and a vector $g=(0,e_{0})^{T}$,
$e_{0}$ is a vector of a standard basis (it has all zero entities
except of zero component which is equal to one). We agree vectors
are the column vectors. Since $a(n)$ has a compact support, $V$
acts on $l_{2}(\mathbb{Z})$ and hence $A$ defines a bounded linear
operator on $L$.

Uniqueness easily follows from the linearity of system (\ref{mainEqMatrixForm}).
Indeed, let $\psi(t)$ and $\psi'(t)$ are two solutions of (\ref{mainEqMatrixForm})
with the same initial condition. Then $\delta(t)=\psi(t)-\psi'(t)$
is a solution of homogeneous equation: 
\[
\dot{\delta}=A\delta
\]
and $\delta(0)=0$ and moreover $\delta(t)\in L$ almost surely for
all $t\geqslant0$. Thus the same arguments as in the classical theory
of ODE in the Banach spaces (see \cite{DalKrein}) show us that $\delta(t)=0$
almost surely for all $t\geqslant0$. So uniqueness has been proved.

A solution of (\ref{mainEqMatrixForm}) can be found via the classical
formula for the solution of inhomogeneous ODE \cite{DalKrein}:
\begin{equation}
\psi(t)=e^{tA}\psi(0)+\sigma\int_{0}^{t}e^{(t-s)A}g_{0}\ dw_{s}=\psi_{0}(t)+\psi_{1}(t),\label{psiSol}
\end{equation}
where $\psi_{0}(t)=e^{tA}\psi(0)$ and $\psi_{1}(t)=\sigma\int_{0}^{t}e^{(t-s)A}g_{0}\ dw_{s}$.
As we mentioned above $A$ is a bounded operator on $L$, and so $e^{tA}$
is correctly defined bounded operator on $L$. Therefore $\psi_{0}(t)\in L$
for all $t\geqslant0$. To prove the same statement with probability
one for $\psi_{1}(t)$ we need the following lemma. \begin{lemma}
\label{exptaformula} For all $t\geqslant0$ the following formula
holds: 
\[
e^{tA}=\left(\begin{array}{cc}
C(t) & S(t)\\
VS(t) & C(t)
\end{array}\right),
\]
where 
\[
C(t)=\sum_{n=0}^{\infty}(-1)^{n}\frac{t^{2n}}{(2n)!}V^{n},\quad S(t)=\sum_{n=0}^{\infty}(-1)^{n}\frac{t^{2n+1}}{(2n+1)!}V^{n}.
\]
\end{lemma}
\begin{proof}
 Straightforward calculation or see at
\cite{DalKrein}. Note that formal series expansion for $\cos tV^{1/2}$
equals to $C(t)$. The same is true for the pair $V^{-1/2}\sin tV^{1/2}$
and $S(t)$.
\end{proof}

Denote $\psi_{1}(t)=(q^{(1)}(t),p^{(1)}(t))^{T}$. From lemma \ref{exptaformula}
we have 
\begin{equation}
q_{k}^{(1)}(t)=\sigma\int_{0}^{t}S_{k,0}(t-s)dw_{s},\quad p_{k}^{(1)}(t)=\sigma\int_{0}^{t}C_{k,0}(t-s)dw_{s}.\label{qpSolInhomSOl}
\end{equation}
We want to prove that $q^{(1)}(t)\in l_{2}(\mathbb{Z})$ and $p^{(1)}(t)\in l_{2}(\mathbb{Z})$
almost surely. The proof will be based on the following lemma.

\begin{lemma} For all $t\geqslant0$ the following inequalities hold:
\begin{equation}
|C_{i,j}(t)|\leqslant\frac{v^{\rho}t^{2\rho}}{(2\rho)!}e^{\sqrt{v}t},\quad|S_{i,j}(t)|\leqslant\frac{v^{\rho}t^{2\rho+1}}{(2\rho+1)!}e^{\sqrt{v}t}, \label{CSineq}
\end{equation}
where $v=||V||_{l_{2}(\mathbb{Z})}$, $\rho=\lceil |i-j| / r\rceil$
and the number $r$ is a radius of interaction, which is defined in
assumption 2 on the function $a$. By $\lceil x\rceil$ we denote
a ceiling function of $x$, i.e.\ the least integer greater than or
equal to $x$.
\end{lemma}
\begin{proof} Since matrix $V$ is translation
invariant, it suffices to prove the assertion for $i=k\geqslant0$
and $j=0$. For all $n\geqslant1$ we have: 
\[
(V^{n})_{k,0}=\sum_{i_{1},\ldots,i_{n-1}\in\mathbb{Z}}V_{i_0,i_{1}}V_{i_{1},i_{2}}\ldots V_{i_{n-1},i_{n}},\ i_0 = k,\ i_n = 0.
\]
We see that if $(V^{n})_{k,0}\ne0$ then $|i_{j}-i_{j-1}|\leqslant r$
for all $j$ and thus we get : 
\[
|k|=|i_0-i_{1}+i_{1}-i_{2}+\ldots +i_{n-1}-i_{n}|\leqslant rn.
\]
And we obtain the inequality: $n\geqslant k / r$. Therefore
the following is true: 
\[
C_{k,0}(t)=\sum_{n\geqslant\frac{k}{r}}^{\infty}(-1)^{n}\frac{t^{2n}}{(2n)!}V^{n}.
\]
Now we derive a bound for $C_{k,0}(t)$: 
\[
|C_{k,0}(t)|\leqslant\sum_{n=\rho}^{\infty}\frac{t^{2n}}{(2n)!}v^{n}=\sum_{n=\rho}^{\infty}\frac{(\sqrt{v}t)^{2n}}{(2n)!}\leqslant\sum_{n=2\rho}^{\infty}\frac{(\sqrt{v}t)^{n}}{n!}\leqslant\frac{v^{\rho}t^{2\rho}}{(2\rho)!}e^{\sqrt{v}t}.
\]
It is easy to see that $S(t)=\int_{0}^{t}C(s)\ ds$ and so the second
inequality in (\ref{CSineq}) follows. \end{proof}

Everything has done to prove the existence and uniqueness lemma \ref{EUlemma}.
Indeed, from (\ref{qpSolInhomSOl}) we obtain : 
\[
\bold{E}(q_{k}^{(1)}(t))^{2}=\sigma^{2}\int_{0}^{t}S_{k,0}^{2}(t-s)\ ds=\sigma^{2}\int_{0}^{t}S_{k,0}^{2}(s)\ ds.
\]
Using inequalities (\ref{CSineq}) we have the bound: 
\[
\bold{E}(q_{k}^{(1)}(t))^{2}\leqslant\sigma^{2}\int_{0}^{t}\left(\frac{v^{\rho}s^{2\rho+1}}{(2\rho+1)!}e^{\sqrt{v}s}\right)\ ds\leqslant\sigma^{2}\frac{v^{\rho}t^{2\rho+2}}{(2\rho+2)!}e^{\sqrt{v}t},
\]
where $\rho=\lceil |k| / r \rceil$. Hence we conclude: 
\[
\sum_{k\in\mathbb{Z}}\bold{E}(q_{k}^{(1)}(t))^{2}<\infty
\]
and due to monotone convergence theorem we get: 
\[
\sum_{k\in\mathbb{Z}}(q_{k}^{(1)}(t))^{2}<\infty
\]
with probability one. Thus $q^{(1)}(t)\in l_{2}(\mathbb{Z})$ almost
sure for all $t\geqslant0$. The similar arguments show that $p^{(1)}(t)\in l_{2}(\mathbb{Z})$
almost sure for all $t\geqslant0$. Thus lemma \ref{EUlemma} is proved.

\subsection{Solution via the Fourier transform}

Consider the Fourier transform of a solution: 
\[
Q_{t}(\lambda)=\sum_{n}q_{n}(t)e^{in\lambda}.
\]
Simple calculation gives us: 
\[
\frac{d^{2}}{dt^{2}}Q_{t}=-\omega^{2}(\lambda)Q_{t}+\sigma\dot{w}_{t}.
\]
Thus $Q_{t}(\lambda)$ for every $\lambda$ satisfies the equation
of harmonic oscillator with frequency $\omega(\lambda)$ influenced
by white noise. Its solution is unique and can easily be found using
standard tools (see \cite{Gitterman,Antonov}) : 
\[
Q_{t}(\lambda)=Q_{t}^{(0)}(\lambda)+Q_{t}^{(1)}(\lambda),
\]
where we denote: 
\begin{align*}
  Q_{t}^{(0)}&=Q_{0}(\lambda)\cos(t\omega(\lambda))+P_{0}(\lambda)\frac{\sin(t\omega(\lambda))}{\omega(\lambda)},\\
  Q_{t}^{(1)}&=\frac{\sigma}{\omega(\lambda)}\int_{0}^{t}\sin((t-s)\omega(\lambda))dw_{s},
\end{align*}
and $P_{0}(\lambda)=\dot{Q}_{0}(\lambda)$. Note that $Q_{t}^{(0)}$
is a solution of the homogeneous equation (with $\sigma=0$) with
initial data $Q_{0}^{(0)}=Q_{0},\ \dot{Q}_{0}^{(0)}=P_{0}$ and $Q_{t}^{(1)}$
is a solution of the inhomogeneous equation with zero initial conditions.

Using the inverse transformation we obtain: 
\[
q_{n}(t)=\frac{1}{2\pi}\int_{0}^{2\pi}e^{-in\lambda}Q_{t}(\lambda)\ d\lambda.
\]
In formula (\ref{psiSol}) we denote $\psi_{k}(t)=(q^{(k)}(t),p^{(k)}(t))^{T}, \ k=0,1$.
It now follows that: 
\begin{equation}
q_{n}^{(k)}(t)=\frac{1}{2\pi}\int_{0}^{2\pi}e^{-in\lambda}Q_{t}^{(k)}(\lambda)\ d\lambda,\quad k=0,1.\label{qknt}
\end{equation}
Thus we almost have proved the following lemma.
\begin{lemma} \label{allSolFormulas}
The following formulas hold: 
\begin{align}
q_{n}^{(0)}(t)= & \ \frac{1}{2\pi}\int_{0}^{2\pi}e^{-in\lambda}\left(Q_{0}(\lambda)\cos(t\omega(\lambda))+P_{0}(\lambda)\frac{\sin(t\omega(\lambda))}{\omega(\lambda)}\right)d\lambda,\label{qz}\\
p_{n}^{(0)}(t)= & \ \frac{1}{2\pi}\int_{0}^{2\pi}e^{-in\lambda}\left(-Q_{0}(\lambda)\omega(\lambda)\sin(t\omega(\lambda))+P_{0}(\lambda)\cos(t\omega(\lambda))\right)d\lambda,\label{pz}\\
q_{n}^{(1)}(t)= & \ \int_{0}^{t}x_{n}(t-s)dw_{s},\quad x_{n}(t)=\frac{\sigma}{2\pi}\int_{0}^{2\pi}e^{-in\lambda}\frac{\sin(t\omega(\lambda))}{\omega(\lambda)}d\lambda,\label{qo}\\
p_{n}^{(1)}(t)= & \ \int_{0}^{t}y_{n}(t-s)dw_{s},\quad y_{n}(t)=\frac{\sigma}{2\pi}\int_{0}^{2\pi}e^{-in\lambda}\cos(t\omega(\lambda))d\lambda .\label{po}
\end{align}
\end{lemma}
\begin{proof} Formula (\ref{qz}) was derived above
at (\ref{qknt}). (\ref{pz}) obtained from (\ref{qz}) by differentiating.
(\ref{qo}) follows from (\ref{qknt}) after switching the order of
integration. We can change the order of integration between It\^o integral
and Lebesgue one because the integrand is a deterministic smooth function
and due to ``integration by parts'' formula 
\[
\int_{0}^{t}f(s)dw_{s}=f(t)w_{t}-\int_{0}^{t}f'(s)w_{s}\ ds
\]
which is true for any smooth function $f$. The last formula (\ref{po})
deduced from (\ref{qo}) and the equality $dq_{n}^{(1)}=p_{n}^{(1)}\ dt$.
\end{proof}

Let us prove here that $U(q)=\sum_{k,j}a(k-j)q_{k}q_{j}\geqslant0$ for all
$q\in l_{2}(\mathbb{Z})$. Using operator $V$ defined above in (\ref{defA})
we can write: 
\[
U(q)=(q,Vq).
\]
Denote $\widehat{f}(\lambda)=\sum_{k}f(k)e^{ik\lambda}$ the Fourier
transform of the sequence $f(k)$. Thus, due to Parseval's theorem
we obtain 
\[
U(q)=\frac{1}{2\pi}\int_{0}^{2\pi}\widehat{q}(\lambda)\overline{\widehat{(Vq)}}(\lambda)d\lambda=\frac{1}{2\pi}\int_{0}^{2\pi}|\widehat{q}(\lambda)|^{2}\omega^{2}(\lambda)d\lambda\geqslant0.
\]
In the last equality we have used an obvious relation $\widehat{(Vq)}(\lambda)=\omega^{2}(\lambda)\widehat{q}(\lambda)$.

\subsection{Global energy behavior}

To prove formula (\ref{EHform}), we need find the expression for
differential $dH$. By definition using It\^o formula we have: 
\[
dp_{k}^{2}=2p_{k}dp_{k}+(dp_{k})^{2}=2p_{k}\Bigl(-\sum_{j}a(k-j)q_{j}dt+\sigma\delta_{k,0}dw_{t} \Bigr)+\frac{\sigma^{2}}{2}\delta_{k,0}dt.
\]
Denote $X_{t}=H(t)=H(\psi(t))$. Whence for the energy we obtain: 
\begin{align*}
dH=dX_{t}&=\frac{1}{2}\sum_{k}dp_{k}^{2}+\frac{1}{2}\sum_{k,j}a(k-j)d(q_{k}q_{j})
\\
         &=-\sum_{k,j}a(k-j)p_{k}q_{j}dt+\frac{1}{2}\sum_{k,j}a(k-j)(p_{k}q_{j}+p_{j}q_{k})dt\\
  &\quad {} +\sigma p_{0}dw_{t}+\frac{\sigma^{2}}{2}dt
\\
&=\frac{\sigma^{2}}{2}dt+\sigma p_{0}dw_{t}.
\end{align*}
This is equivalent to the equality: 
\begin{align*}
  H(t)&=H(0)+\frac{\sigma^{2}}{2}t+\sigma\int_{0}^{t}p_{0}(s)dw_{s}\\
  &=H(\psi(0))+\frac{\sigma^{2}}{2}t+\sigma\int_{0}^{t}p_{0}^{(0)}(s)+p_{0}^{(1)}(s)\ dw_{s}.
\end{align*}
In the last equality we have used (\ref{psiSol}). Substituting (\ref{pz})
and (\ref{po}) into the last expression we get (\ref{energyDecomp}).

Formulas (\ref{EHform}) and (\ref{HvarFor}), for the expected value
and variance of the energy $H(t)$ immediately follows from (\ref{energyDecomp}).
Now we prove equality (\ref{HvarAsymp}). From (\ref{HvarFor}) we
have: 
\[
\bold{D}H(t)=t^{2}\int_{0}^{1}(1-s)^{2}h^{2}(ts)\ ds.
\]
Lemma \ref{202001122342} gives us: 
\[
\lim_{T\rightarrow\infty}\frac{1}{T}\int_{0}^{T}h^{2}(s)ds=0.
\]
Therefore due to Bochner's theorem (see \cite{Kawata}, p. 182, th.
5.5.1) the following limit holds: 
\[
\int_{0}^{1}(1-s)^{2}h^{2}(ts)\ ds\rightarrow0,\ \mbox{as}\ t\rightarrow\infty.
\]
So (\ref{HvarAsymp}) has proved.

Note that our derivation of (\ref{energyDecomp}) has some disadvantages
in a place where we sum up the infinite number of It\^o differentials.
It is not hard to legitimize this procedure using a corresponding
limit or one can use a more direct approach based on lemma \ref{allSolFormulas}.

Now let us prove (\ref{ETform}). From presentation (\ref{psiSol})
and lemma \ref{allSolFormulas} we have: 
\[
\bold{E}T(t)=\frac{1}{2}\sum_{n}(p_{n}^{(0)}(t))^{2}+\frac{1}{2}\sum_{n}\bold{E}(p_{n}^{(1)}(t))^{2}=O(1)+\frac{1}{2}\sum_{n}\int_{0}^{t}y_{n}^{2}(s)ds.
\]
The last equality is due to It\^o isometry. Next we study the sum. Note
that $y_{n}(t)$ is a Fourier coefficient of the function $\cos(t\omega(\lambda))$.
Hence using the Parseval's theorem we obtain: 
\[
\sum_{n}\int_{0}^{t}y_{n}^{2}(s)ds=\int_{0}^{t}\sum_{n}y_{n}^{2}(s)ds=\sigma^{2}\int_{0}^{t}\frac{1}{2\pi}\int_{0}^{2\pi}\cos^{2}(s\omega(\lambda))d\lambda ds=
\]
\[
=\frac{\sigma^{2}t}{2}+\frac{1}{2}\int_{0}^{t}\frac{1}{2\pi}\int_{0}^{2\pi}\cos(2s\omega(\lambda))d\lambda ds.
\]
In the last equality we have used a school formula $\cos^{2}x= [1+\cos(2x)] / 2$.
Lemma \ref{202001122342} gives us: 
\[
\frac{1}{2}\int_{0}^{t}\frac{1}{2\pi}\int_{0}^{2\pi}\cos(2s\omega(\lambda))d\lambda ds=O(t^{1-\varepsilon})
\]
for some $\varepsilon>0$. Thus we have proved the formula for the
mean kinetic energy (\ref{ETform}). Equality (\ref{EUform}) immediately
follows from (\ref{EHform}) and (\ref{ETform}) because of relation
$H=T+U$. This completes the proof of Theorem \ref{globalEnergyBehTh}.

\subsection{Local energy asymptotics}

Now we prove Theorem \ref{locEnergyTheorem}. Let us begin with kinetic
energy: 
\[
T_{n}(t)=\frac{p_{n}^{2}(t)}{2}.
\]
Remark that according to formulas (\ref{qz}),(\ref{pz}) and Riemann\tire Lebesgue
lemma we have the limit: 
\begin{equation}
\lim_{t\rightarrow\infty}\psi_{0}(t)=0.\label{psizerolim}
\end{equation}
Hence we obtain: 
\begin{align*}
  \bold{E}T_{n}(t)&=\frac{1}{2}\bold{E}\left((p_{n}^{(0)}(t))^{2}+(p_{n}^{(1)}(t))^{2}+2p_{n}^{(0)}(t)p_{n}^{(1)}(t)\right)\\
  &=\frac{1}{2}\left((p_{n}^{(0)}(t))^{2}+\bold{E}(p_{n}^{(1)}(t))^{2}\right)
\\
&=\frac{1}{2}\bold{E}(p_{n}^{(1)}(t))^{2}+\bar{\bar{o}}(1),\quad  \mbox{as}\ t\rightarrow\infty.
\end{align*}
The application of (\ref{po}) and It\^o isometry yields the equality:
\begin{equation}
\bold{E}(p_{n}^{(1)}(t))^{2}=\int_{0}^{t}y_{n}^{2}(t-s)ds=\int_{0}^{t}y_{n}^{2}(s)ds.\label{202001091457}
\end{equation}
Since $\omega(\lambda)$ is even function,  for $y_{n}(t)$ we
get: 
\[
y_{n}(t)=\frac{\sigma}{2\pi}\int_{0}^{2\pi}e^{-in\lambda}\cos(t\omega(\lambda))d\lambda=\frac{\sigma}{2\pi}\int_{0}^{2\pi}\cos(n\lambda)\cos(t\omega(\lambda))d\lambda.
\]
Lemma \ref{eftlemma} gives us: 
\begin{align*}
  y_{n}(t)&=\sigma\frac{1}{\sqrt{t}}\sum_{k=0}^{m+1}\theta_{k}\sqrt{\frac{1}{2\pi|\omega''(\lambda_{k})|}}\cos(n\lambda_{k})
            \cos\Bigl(t\omega(\lambda_{k})+\frac{\pi}{4}s(\lambda_{k})\Bigr)+O\Bigl(\frac{1}{t}\Bigr)
\\
&=\frac{1}{\sqrt{t}}\sum_{k=0}^{m+1}b_{k}u_{k}(t)+O\Bigl(\frac{1}{t}\Bigr),
\end{align*}
where we denote: 
\[
b_{k}=\sigma\theta_{k}\sqrt{\frac{1}{2\pi|\omega''(\lambda_{k})|}}\cos(n\lambda_{k}),\ u_{k}(t)=\cos\Bigl(t\omega(\lambda_{k})+\frac{\pi}{4}s(\lambda_{k})\Bigr)
\]
For the square we get: 
\[
y_{n}^{2}(t)=\frac{1}{t}\sum_{k=0}^{m+1}b_{k}^{2}u_{k}^{2}(t)+\frac{1}{t}\sum_{k\ne j}b_{k}b_{j}u_{k}(t)u_{j}(t)+O\left(\frac{1}{t\sqrt{t}}\right).
\]
Substitute the last expression for $y_{n}^{2}(t)$ to (\ref{202001091457}): 
\begin{align}
\bold{E}(p_{n}^{(1)}(t))^{2}&=\int_{0}^{1}y_{n}^{2}(s)ds+\int_{1}^{t}y_{n}^{2}(s)ds=O(1)+\int_{1}^{t}y_{n}^{2}(s)ds
\nonumber \\
&=O(1)+\sum_{k=0}^{m+1}b_{k}^{2}\int_{1}^{t}\frac{u_{k}^{2}(s)}{s}ds+\sum_{k\ne j}b_{k}b_{j}\int_{1}^{t}\frac{u_{k}(s)u_{j}(s)}{s}ds.\label{202001091523}
\end{align}
Using the formula $\cos^{2}x=[1+\cos(2x)]/2$ we obtain: 
\begin{align*}
  \int_{1}^{t}\frac{u_{k}^{2}(s)}{s}ds&=\int_{1}^{t}\frac{\cos^{2}(s\omega(\lambda_{k})+\frac{\pi}{4}s(\lambda_{k}))}{s}ds\\
  &=\frac{\ln t}{2}+\frac{1}{2}\int_{1}^{t}\frac{\cos(2s\omega(\lambda_{k})+\frac{\pi}{2}s(\lambda_{k}))}{s}ds
=\frac{1}{2}\ln t+O(1).
\end{align*}
The last equality follows from the fact that integrals: 
\begin{equation}
\mathrm{ci}(1)=\int_{1}^{+\infty}\frac{\cos x}{x}dx<\infty,\quad\mathrm{si}(1)=-\int_{1}^{+\infty}\frac{\sin x}{x}dx<\infty\label{cisiconv}
\end{equation}
are converged (see \cite{GR} p.\thinspace 656, 3.721) and that $\omega(\lambda_{k})>0$
for all $k=0,\ldots,m+1$ due to assumption A1). To estimate the
remainder term in (\ref{202001091523}) we use the formula $\cos(a)\cos(b)=\frac{1}{2}(\cos(a+b)+\cos(a-b))$:
\begin{align*}
\int_{1}^{t}\frac{u_{k}(s)u_{j}(s)}{s}ds&=\frac{1}{2}\int_{1}^{t}\frac{\cos(s(\omega(\lambda_{k})+\omega(\lambda_{j}))+\frac{\pi}{4}(s(\lambda_{k})+s(\lambda_{j})))}{s}ds
\\
&\quad {} +\frac{1}{2}\int_{1}^{t}\frac{\cos(s(\omega(\lambda_{k}) \! - \! \omega(\lambda_{j}))+\frac{\pi}{4}(s(\lambda_{k}) \! - \! s(\lambda_{j})))}{s}ds=O(1).
\end{align*}
The last equality is derived from (\ref{cisiconv}) using assumption A3).
Substituting it to (\ref{202001091523}) we obtain: 
\[
\bold{E}(p_{n}^{(1)}(t))^{2}=O(1)+\frac{\ln(t)}{2}\sum_{k=0}^{m+1}b_{k}^{2}.
\]
Thus the equality (\ref{kineticEnergyAsym}) has proved.

Now prove equality (\ref{potentialEnergyAsym}) of theorem \ref{locEnergyTheorem}.
The idea of the proof is the same as for the kinetic energy. From
 (\ref{psizerolim}) we obtain: 
\begin{equation}
\bold{E}U_{n}(t)=\frac{1}{2}\sum_{j}a(n-j)\bold{E}(q_{n}^{(1)}(t)q_{j}^{(1)}(t))+\bar{\bar{o}}(1),\ \mbox{as}\ t\rightarrow\infty.\label{eunteq}
\end{equation}
Equality (\ref{qo}) and It\^o isometry give us: 
\begin{align}
  \bold{E}(q_{n}^{(1)}(t)q_{j}^{(1)}(t))&=\int_{0}^{t}x_{n}(t-s)x_{j}(t-s)ds\nonumber \\
  &=\int_{0}^{t}x_{n}(s)x_{j}(s)ds=\int_{1}^{t}x_{n}(s)x_{j}(s)ds+O(1).\label{eqqeq}
\end{align}
Since $\frac{\sin(t\omega(\lambda))}{\omega(\lambda)}$ is an even
function and due to lemma \ref{eftlemma}, we have for all $n\in\mathbb{Z}$:
\begin{align*}
x_{n}(t)&=\frac{\sigma}{2\pi}\int_{0}^{2\pi}\cos(n\lambda)\frac{\sin(t\omega(\lambda))}{\omega(\lambda)}d\lambda
\\
&=\sigma\frac{1}{\sqrt{t}}\sum_{k=0}^{m+1}\theta_{k}\sqrt{\frac{1}{2\pi|\omega''(\lambda_{k})|}}\frac{\cos(n\lambda_{k})}{\omega(\lambda_{k})}\sin(t\omega(\lambda_{k})+\frac{\pi}{4}s(\lambda_{k}))+O\left(\frac{1}{t}\right)
\\
&=\frac{1}{\sqrt{t}}\sum_{k=0}^{m+1}e_{k}^{(n)}v_{k}(t)+O\left(\frac{1}{t}\right),
\end{align*}
where we denote: 
\[
  e_{k}^{(n)}=\sigma\theta_{k}\sqrt{\frac{1}{2\pi|\omega''(\lambda_{k})|}}\frac{\cos(n\lambda_{k})}{\omega(\lambda_{k})},\quad
  v_{k}(t)=\sin(t\omega(\lambda_{k})+\frac{\pi}{4}s(\lambda_{k})).
\]
Substitute the last expression to (\ref{eqqeq}): 
\begin{align*}
  \bold{E}(q_{n}^{(1)}(t)q_{j}^{(1)}(t))&=\sum_{k=0}^{m+1}e_{k}^{(n)}e_{k}^{(j)}\int_{1}^{t}\frac{v_{k}^{2}(s)}{s}ds\\
  &\quad {} +\sum_{k_{1}\ne k_{2}}e_{k_{1}}^{(n)}e_{k_{2}}^{(j)}\int_{1}^{t}\frac{v_{k_{1}}(s)v_{k_{2}}(s)}{s}ds+O(1).
\end{align*}
The same arguments as in the case of the kinetic energy give us equalities:
\[
\int_{1}^{t}\frac{v_{k}^{2}(s)}{s}ds=\frac{1}{2}\ln t+O(1),\quad\int_{1}^{t}\frac{v_{k_{1}}(s)v_{k_{2}}(s)}{s}=O(1)
\]
if $k_1 \ne k_2$. Therefore we have 
\[
\bold{E}(q_{n}^{(1)}(t)q_{j}^{(1)}(t))=\frac{1}{2}\ln t\sum_{k=0}^{m+1}e_{k}^{(n)}e_{k}^{(j)}+O(1).
\]
Put this expression to formula (\ref{eunteq}). Then we obtain: 
\[
\bold{E}U_{n}(t)=D_{n}\ln t+O(1),\quad D_{n}=\frac{1}{4}\sum_{j}a(n-j)\sum_{k=0}^{m+1}e_{k}^{(n)}e_{k}^{(j)}.
\]
Now we prove that $D_{n}=d_{n}$ where $d_{n}$ defined in (\ref{dneqdef}).
At the first step we change the summation order: 
\[
D_{n}=\frac{1}{4}\sum_{k=0}^{m+1}e_{k}^{(n)}\sum_{j}a(n-j)e_{k}^{(j)}.
\]
For the internal sum we get: 
\[
\sum_{j}a(n-j)e_{k}^{(j)}=\sigma\theta_{k}\sqrt{\frac{1}{2\pi|\omega''(\lambda_{k})|}}\frac{1}{\omega(\lambda_{k})}\sum_{j}a(n-j)\cos(j\lambda_{k}).
\]
The simple algebra shows us: 
\[
\sum_{j}a(n-j)\cos(j\lambda_{k})=\sum_{j}a(n-j)\frac{e^{ij\lambda_{k}}+e^{-ij\lambda_{k}}}{2}=\cos(n\lambda_{k})\omega^{2}(\lambda_{k}).
\]
Whence we have: 
\[
\sum_{j}a(n-j)e_{k}^{(j)}=\omega^{2}(\lambda_{k})e_{k}^{(n)}.
\]
Thus we obtain: 
\[
D_{n}=\frac{1}{4}\sum_{k=0}^{m+1}\omega^{2}(\lambda_{k})(e_{k}^{(n)})^{2}=d_{n}.
\]
This completes the proof of Theorem \ref{locEnergyTheorem}.

\begin{lemma} \label{202001122342} There are positive constants
$b,\varepsilon$ such that for all sufficiently large $t$ the following
inequality holds: 
\begin{equation}
\left|\int_{0}^{2\pi}e^{it\omega(\lambda)}d\lambda\right|\leqslant bt^{-\varepsilon}.\label{202001122115}
\end{equation}
\end{lemma} \begin{proof} Recall that 
\[
\omega^{2}(\lambda)=a(0)+2\sum_{n=1}^{r}a(n)\cos(n\lambda).
\]
If $\omega(\lambda)$ is strictly greater than zero (i.e.\ assumption
A1 holds) then lemma immediately follows from the stationary phase
method. Indeed in that case $\omega(\lambda)$ is an analytic function
and by the stationary phase method the asymptotic of the integral
at (\ref{202001122115}) is determined by stationary points of $\omega(\lambda)$
(see \cite{Erdelyi,Fedoruk}). Since $\omega(\lambda)$ is an analytic
then $\omega(\lambda)$ has a finite number of critical points on
$[0,2\pi]$ and each of this point has finite multiplicity. Hence
the inequality (\ref{202001122115}) follows from the corresponding
asymptotic formulas of the stationary phase method.\par Now suppose
that $\omega(\lambda)$ has zeros on $[0,2\pi]$. Denote $f(\lambda)=\omega^{2}(\lambda)$.
Since $f$ is an analytic, $\omega(\lambda)$ has a finite number
of zeros on $[0,2\pi]$. Consider some zero $z\in[0,2\pi]$ of $\omega(\lambda)$
and study integral (\ref{202001122115}) over a small neighborhood
of $z$. From analyticity of $f$ follows that there is a number $n\geqslant1$
such that: 
\[
f(z)=0,\ f'(z)=0,\ldots f^{(n-1)}(z)=0,\ f^{(n)}(z)\ne0.
\]
Since $f$ is non-negative, $n=2m$ is an even number. This is a well-known
fact (see \cite{Erdelyi,Fedoruk}) that in the given case there is
a $C^{\infty}$-smooth one-to-one function $\varphi(y)$ mapping some
neighborhood of zero, say $[-\delta,\delta]$, to a small vicinity
of $z$, denote it by $[z-\delta',z+\delta']$, such that 
\[
f(\varphi(y))=y^{n},\quad  \varphi(0)=z.
\]
Therefore for the integral we have 
\begin{align*}
  \int_{z-\delta'}^{z+\delta'}e^{it\omega(\lambda)}d\lambda&=\int_{-\delta}^{\delta}\exp(it|y|^{n/2})\varphi'(y)dy\\
  &=\int_{0}^{\delta}\exp(ity^{m})(\varphi'(y)+\varphi'(-y))dy=O\bigl(t^{-1/m}\bigr).
\end{align*}
Thus we have proved that for each zero $z$ of $f$ there is a neighborhood
of $z$ such that the integral over this neighborhood satisfies inequality
(\ref{202001122115}) . The same statement is evidently true for the
critical points of $f$. Hence, splitting integral (\ref{202001122115})
into the integrals over such neighborhoods and remaining part without
zeros and critical points, we get the proof of (\ref{202001122115}).
\end{proof}

\begin{lemma} \label{eftlemma} Consider the integral 
\[
E_{f}(t)=\int_{0}^{2\pi}g(\lambda)e^{it\omega(\lambda)}d\lambda,\ t\geqslant0
\]
for some $2\pi$-periodic real-valued $C^{\infty}(\mathbb{R})$-smooth
even function $g$. Then under the assumptions A1) and A2) the following
formula holds: 
\begin{equation}
E_{f}(t)=\frac{1}{\sqrt{t}}\sum_{k=0}^{m+1}\theta_{k}\sqrt{\frac{2\pi}{|\omega''(\lambda_{k})|}}g(\lambda_{k})e^{it\omega(\lambda_{k})+\frac{i\pi}{4}s(\lambda_{k})}+O\left(\frac{1}{t}\right)\label{efasym}
\end{equation}
where 
\[
\theta_{k}=\begin{cases}
2, & k=1,\ldots,m,\\
1, & k\in\{0,m+1\},
\end{cases} \quad s(\lambda)=\mathrm{sgn}(\omega''(\lambda)),
\]
and $\lambda_{0},\ldots,\lambda_{m+1}$ are critical points of the
function $\omega(\lambda)$ introduced in assumption A2.
\end{lemma}
\begin{proof} We will use the stationary phase method (see \cite{Erdelyi,Fedoruk}).
Note that $\omega(\lambda)=\omega(2\pi-\lambda)$ for all $\lambda$.
Hence the only critical points of $\omega(\lambda)$ on the interval
$[0,2\pi)$ are $\lambda_{0},\ldots,\lambda_{m+1}$ and $\mu_{1},\ldots,\mu_{m}$
where $\mu_{j}=2\pi-\lambda_{j},\ j=1,\ldots,m$. Recall that $\lambda_0=0$
and we want to shift the interval of integration from the boundary
stationary point. Since functions $g$ and $\omega$ are $2\pi$ periodic,
we can write 
\[
E_{f}(t)=\int_{-\delta}^{2\pi-\delta}g(\lambda)e^{it\omega(\lambda)}d\lambda,
\]
where we choose small number $\delta$ in such a way all critical
points $\lambda_{0},\ldots,\lambda_{m+1}$ and $\mu_{1},\ldots,\mu_{m}$
lie strongly inside the interval $(-\delta,2\pi-\delta)$. By stationary
phase method we have the asymptotic formula: 
\begin{align*}
  E_{f}(t)&\sim\frac{1}{\sqrt{t}}\sum_{k=0}^{m+1}\sqrt{\frac{2\pi}{|\omega''(\lambda_{k})|}}g(\lambda_{k})
    \exp\Bigl\{it\omega(\lambda_{k})+\frac{i\pi}{4}s(\lambda_{k})\Bigr\} \\
  &\quad {} +\frac{1}{\sqrt{t}}\sum_{k=1}^{m}\sqrt{\frac{2\pi}{|\omega''(\mu_{k})|}}g(\mu_{k}) \exp\Bigl\{it\omega(\mu_{k})+\frac{i\pi}{4}s(\mu_{k})\Bigr\}.
\end{align*}
Since $\omega(\mu_{k})=\omega(\lambda_{k}),\ \omega''(\mu_{k})=\omega''(\lambda_{k}),\ g(\mu_{k})=g(\lambda_{k})$,
we obtain the leading term in (\ref{efasym}). The term $O(t^{-1})$ comes
from contribution of the boundary points.
\end{proof}

\end{document}